\documentclass[12pt,a4paper,normalheadings,headsepline,headinclude,bibtotoc,fleqn,review]{article}
\usepackage[latin1]{inputenc}
\usepackage{amsmath}
\usepackage{booktabs}
\usepackage{array}
\usepackage{amsthm}
\usepackage{dcolumn}
\usepackage{endnotes}
\usepackage{units}
\usepackage{marvosym}
\usepackage[lines=10]{geometry}
\usepackage{longtable}
\usepackage{rotating}
\usepackage{expdlist}
\usepackage{geometry}
\usepackage{floatrow}
\usepackage{pgfplots}
\usepackage{tikz}
\usetikzlibrary{shapes,snakes}
\pgfmathdeclarefunction{gauss}{3}{
	\pgfmathparse{1/(#3*sqrt(2*pi))*exp(-((#1-#2)^2)/(2*#3^2))}%
}
\usetikzlibrary{decorations.pathreplacing,calligraphy,backgrounds}
\pgfplotsset{compat=1.16}

\pgfmathsetmacro\valueA{gauss(2.25, 2.25, 0.7)} 
 
\usepackage{wrapfig}
\usepackage{expdlist}
\usepackage{amssymb}
\usepackage{gloss}
\usepackage{nomencl}
\usepackage{natbib}
\newtheorem{rem}{Remark}
\newtheorem{lemma}{Lemma}
\newtheorem{prop}{Proposition}
\newtheorem{coro}{Corollary}
\newtheorem{definition}{Definition}

\makegloss
 
\renewcommand{\thesection}{\arabic{section}}

\usepackage{fancyhdr} 
\begin{document}
 
\begin{center}\huge{Equilibrium Selection in Coordination Games with Planned Actions and Scouting}\footnote{This paper is a completely rewritten version of \citet{kuhle2024gamesplannedactionsscouting}. I am grateful to Ariel Rubinstein and an anonymous associate editor for several very helpful suggestions, including shifting the framework from matching pennies to a coordination game and focusing on equilibrium selection. I thank Dominik Grafenhofer, Manuel Mago, and Gregor von Schweinitz for comments on the earlier draft. \\ \textbf{Declarations:} This research received no external funding, does not use any data, and there are no competing interests. AI was used to check language and grammar.}
\end{center}
 
\begin{center} \textit{Wolfgang Kuhle}\\ \textit{Corvinus University of Budapest, Hungary, E-mail wkuhle@gmx.de\\ MEA, Max Planck Institute for Social Law and Social Policy, Munich, Germany}\end{center}
 
\noindent\emph{\textbf{Abstract:} We study Stag Hunt games in which every action requires planning and preparation. Before players act, they observe noisy signals of their opponent's preparations and can revise their plans, which allows players to screen for cooperation. As screening becomes more precise and cheap, small trembles ensure that, in every equilibrium, play yields the payoff-dominant outcome with a probability that converges to one, even when the safe, low payoff equilibrium is risk dominant. Across scenarios, we emphasize that decomposing an action into (i) preparation and (ii) execution allows us to analyze simultaneous-move games in which players partially observe each other's contemporaneous actions.}

\section{Introduction}\label{Introduction}
 
One important clue to players' future actions is their preparations. Such preparations are visible when crowds of depositors wait for a bank to open or when soldiers monitor allied troop movements before a coordinated attack. Likewise, in team efforts such as stag hunts, players can observe whether their counterparts are equipped to hunt rabbits or stags.
 
To incorporate the preparations that players undertake before they act, the present paper studies games where play involves three stages. In the first stage, both players choose which action to prepare. In the second stage, both players receive noisy information about their opponent's preparations. Finally, given the information about their opponent's planned action and their own planned action, players choose whether to execute or revise their plans. Revisions delay payoffs and are thus costly. Once neither player chooses to make further revisions, plans are executed and payoffs materialize.
 
The present paper focuses on stag-hunt-type coordination games. Within the stag hunt fable, the three stage model can be interpreted as (i) players prepare either rabbit nets or stag-hunting gear; (ii) they inspect each other's equipment from a distance; and (iii) a player who detects a mismatch can deliberately make noise to spook rabbits and stags alike, aborting the hunt and forcing a reset to the preparation stage.
 
\textit{Organization:} Section \ref{Section Related Literature} reviews the related literature. Section \ref{Primitive Stag Hunt Game} introduces the baseline stag hunt game, and Section \ref{Planned Actions and Scouting} embeds this game within our preparation-and-scouting protocol. Section \ref{Pure Strategy Equilibria with Planned Actions} shows that both pure-strategy equilibria survive. Section \ref{Marginalizing Inferior Equilibria} shows that, as signals become precise and revisions cheap, players' ability to screen for cooperation pushes mixed equilibria toward the safe equilibrium $(B,B)$. In turn, we show that the $(B,B)$ equilibrium is unstable and does not survive small exogenous trembles. Hence, with trembles, every equilibrium ends in the payoff-dominant outcome $(A,A)$ with a probability that converges to one, even when $(B,B)$ is risk dominant. Section \ref{Section Basins} provides diagrammatic intuition. Section \ref{Section Unilateral Execution} shows that our selection result relies on each player's ability to veto executions. That is, once players can force the execution of plans at no cost, all equilibria of the stage game reappear. On the contrary, once forcing executions is sufficiently costly, equilibrium selection is restored. Section \ref{Discussion} concludes.
 
\section{Related Literature}\label{Section Related Literature}
The present paper relates to two strands of the literature: (i) games with asymmetric information and (ii) conjectural equilibria.\footnote{See, e.g., \citet{Hah77}, \citet{Rub94}, \citet{Bat97}, and \citet{Ang06}. The literature on learning in dynamic games \citep[e.g.,][]{Cha99, Fra12} similarly emphasizes that players can only learn from each other's past actions.} Among others, \citet{Ang06} (p. 1730) argue that it is ``naive'' to assume that players can observe each other's actions in simultaneous-move games since an action must have already occurred to be observable. To circumvent this issue, \citet{Rub94} and \citet{grafenhofer2022observing} restrict attention to ``steady states'' in which players observe and repeat the actions that other players took in the past. Our model instead decomposes actions into preparation and execution. Because players monitor each other's preparations, they receive partially revealing signals about their opponent's impending moves, which resolves the ``chicken-and-egg'' paradox between signals and actions in simultaneous-move games.
 
 
The present paper also contributes to the literature on how asymmetric information over fundamentals and different forms of communication affect equilibrium play \citep{Rub89, Car93, Bin01, Izm10, Ste11, Ber16, Kuh15, Gra16}. In these articles, players use (i) their private information and (ii) knowledge of the equilibrium to \textit{indirectly} infer each other's actions. In our model, players learn about each other's actions \textit{directly}, through noisy signals of their preparations.

A preplay procedure related to ours is studied by \citet{KALAI1981375}, who allows players to revise publicly announced strategies over a fixed number of discrete rounds, after which the final profile is executed. Within each round, players who revise are locked into their revised strategies, which are announced before the remaining players decide how to respond. The present model differs in four respects: (i) preparations are only partially observed through noisy signals rather than perfectly observed, (ii) revisions are costly rather than costless, (iii) execution is endogenous, triggered by mutual consent to stop revising, or, in Section \ref{Section Unilateral Execution}, by one player unilaterally forcing execution, rather than occurring automatically after a predetermined number of rounds, and (iv) revisions are simultaneous and flexible rather than locked in. The fourth difference is the important one. In Kalai's procedure, a player who switches to the cooperative action commits to it, and the opponent can respond to the observed switch. In our model, no player can commit to a preparation, and the opponent's consent may depend on a noisy signal of it. Instead, selection arises from screening: players use their veto over execution to wait for evidence of cooperative preparations, which exogenous trembles provide.
 
Related approaches let players learn about each other's intentions before play. In \citet{bachi2013betrayal}, players observe each other's strategies unless the opponent engages in costly deception; in \citet{BlockLevine2016}, players receive noisy signals about each other's ``codes of conduct'' in repeated games. In both frameworks, observed strategies act as conditional commitment devices, and the main results are folk theorems: observing intentions enlarges the set of equilibrium outcomes. Our results point in the opposite direction: because players observe prepared actions rather than strategies and cannot commit to a reaction function, observing reversible planned actions selects $(A,A)$ instead of enlarging the set of equilibria.
 
A related literature obtains selection from asynchronous timing. In \citet{LagunoffMatsui1997}, players of a repeated pure coordination game move in alternation, and for sufficiently patient players the Pareto-dominant payoff is the unique equilibrium payoff. \citet{Dutta2012} shows that asynchronicity together with a finite horizon selects the Pareto-dominant action profile in games of common interest. Closest to our preparation stage are revision games, in which players prepare actions that are implemented at a predetermined deadline and receive stochastic opportunities to revise them \citep{KamadaKandori2020}. When revision opportunities arrive asynchronously, \citet{CalcagnoKamadaLovoSugaya2014} show that the best action profile of a two-player common interest game is selected. As with the lock-in in \citet{KALAI1981375}, selection in these models rests on commitment: a player who has prepared the cooperative action cannot revise it before his next opportunity, so that the opponent can safely follow.
 
Taken together, the procedures discussed above transform the simultaneous-move stage game into a sequential one, in which a player can observe and respond to an opponent's committed move. Our model preserves simultaneity: preparations are chosen simultaneously, remain revisable, and are observed only through noisy signals, so that no player can credibly lead.
 
\section{Stag Hunt Game}\label{Primitive Stag Hunt Game}
In this section we recall the payoffs and equilibria of the textbook Stag Hunt game.
\begin{table}[h!]
\centering
\renewcommand{\arraystretch}{1.5}
\begin{tabular}{c|c|c|}
\multicolumn{1}{c}{} & \multicolumn{1}{c}{Stag $A$} & \multicolumn{1}{c}{Hare $B$} \\ \cline{2-3}
Stag $A$ & $a, a$ & $c, d$ \\ \cline{2-3}
Hare $B$ & $d, c$ & $b, b$ \\ \cline{2-3}
\end{tabular}
\caption{General Stag Hunt Matrix with Actions $A$ and $B$ and payoffs $a > d > b > c>0$.}
\end{table}
 
\begin{prop}\label{Prposition 1 Primitive Stag Hunt}
There exist two pure strategy equilibria: $(A, A)$ is the payoff-dominant equilibrium and  $(B, B)$ is the safe low-payoff equilibrium. There also exists a symmetric mixed-strategy equilibrium, where each player chooses action $A$ with probability $\alpha^* = \frac{b - c}{(a - d) + (b - c)}$ and action $B$ with probability $1 - \alpha^*$.
\end{prop} 
 
\section{Planned Actions and Scouting}\label{Planned Actions and Scouting}
Once actions are decomposed into preparation and execution, play unfolds according to the following protocol:
\begin{enumerate}
   \item \textit{Planning/Preparation Stage:} Players $i=1,2$ choose whether to prepare action $A_i$ or $B_i$. Let $\alpha_i$ denote the probability that player $i$ prepares action $A_i$ and $1-\alpha_i$ the probability of preparing $B_i$.
 
    \item \textit{Intelligence and Scouting Stage:} After plans are chosen, each player receives a noisy signal about the opponent's planned action. Player 1 receives signal $\hat{A}_1$ or $\hat{B}_1$ about Player 2's plan, which is correct with probability $p \in(1/2,1)$:
    \[
    P(\hat{A}_1 \mid A_2) = P(\hat{B}_1 \mid B_2) = p, \quad P(\hat{B}_1 \mid A_2) = P(\hat{A}_1 \mid B_2) = 1-p
    \]
    Similarly, Player 2 receives signals $\hat{A}_2$ or $\hat{B}_2$ about Player 1's plan, which are correct with probability $q \in(1/2,1)$:
    \[
    P(\hat{A}_2 \mid A_1) = P(\hat{B}_2 \mid B_1) = q, \quad P(\hat{B}_2 \mid A_1) = P(\hat{A}_2 \mid B_1) = 1-q
    \] Signals are drawn independently across players and iterations. We write $\pi_i$ for the precision of Player $i$'s signal, i.e., $\pi_1:=p$ and $\pi_2:=q$.
 
    \item \textit{Execution/Revision of Plans:} Given their own planned action and the received signal, each player decides whether to execute or revise their plan. If at least one player chooses to revise, the game returns to the Planning/Preparation Stage. If neither player wishes to revise further, the planned actions are executed and payoffs are realized.\\
     Players' time preferences are described by $\delta\in(0,1)$: if plans are executed after $t$ revisions, payoffs are discounted by $\delta^t$, and if plans are never executed, players receive zero.
\end{enumerate} 
The protocol defines a dynamic game of imperfect information: when deciding whether to execute, each player knows his own preparation but observes the opponent's preparation only through a noisy signal. To study this game, we use:
 
\begin{definition} A (stationary) strategy for Player $i \in \{1,2\}$ is a pair $s_i = (\alpha_i, \chi_i)$, where:
\begin{enumerate}
    \item $\alpha_i \in [0,1]$ is the probability of preparing action $A_i$ in Stage 1.
    \item $\chi_i: \{A_i, B_i\} \times \{\hat{A}_i, \hat{B}_i\} \to \{0,1\}$ specifies the decision to revise $\chi_i=0$ or execute $\chi_i=1$ in Stage 3 given preparation $x_i \in \{A_i, B_i\}$ and observed signal $\hat{y}_i \in \{\hat{A}_i, \hat{B}_i\}$.
\end{enumerate}
\end{definition}
 
In this notation, the current plans $(x_1,x_2)$ are executed if $\min\{\chi_1,\chi_2\}=1$, and the game returns to Stage 1 otherwise.

\begin{definition}\label{Def PBE} A Perfect Bayesian Equilibrium (PBE) for the game with planned actions and scouting consists of a strategy profile $s^* = (s_1^*, s_2^*)$ and a belief system $\mu = (\mu_1, \mu_2)$ such that:
\begin{enumerate}
    \item \textbf{Belief Consistency:} For each Player $i \in \{1,2\}$, upon choosing preparation $x_i \in \{A_i, B_i\}$ and observing signal $\hat{y}_i \in \{\hat{A}_i, \hat{B}_i\}$, the posterior belief $\mu_i(x_{-i} \mid x_i, \hat{y}_i)$ regarding the opponent's prepared action $x_{-i}$ is derived from the opponent's planning strategy $\alpha_{-i}^*$ and the signal precisions ($p$ or $q$) using Bayes' rule.
    \item \textbf{Sequential Rationality in Stage 3:} For any preparation $x_i$ and observed signal $\hat{y}_i$, the execution decision $\chi_i^*(x_i, \hat{y}_i)$ maximizes Player $i$'s expected payoff (either immediate execution or the discounted continuation value of revising), given the belief $\mu_i(\cdot \mid x_i, \hat{y}_i)$ and the opponent's strategy $s_{-i}^*$.
    \item \textbf{Optimality in Stage 1:} The preparation probability $\alpha_i^*$ maximizes Player $i$'s expected ex-ante payoff, anticipating the subsequent signal distribution and optimal execution choices of both players.
\end{enumerate}
\end{definition}
 
\section{Pure Strategy Equilibria}\label{Pure Strategy Equilibria with Planned Actions}
 
\noindent Before we study the conditions that select the payoff dominant equilibrium, we establish the pure-strategy benchmarks of the game without trembles.
 
\begin{prop}\label{Proposition Pure Equilibria}
The modified Stag Hunt game with scouting admits the following two pure-strategy Perfect Bayesian Equilibria:
\begin{enumerate}
    \item \textit{A payoff-dominant equilibrium in which both players plan action $A$ and execute after either signal, yielding the payoff vector $(a,a)$.}
    \item \textit{A low-payoff equilibrium in which both players plan action $B$ and execute after either signal, yielding the payoff vector $(b,b)$.}
\end{enumerate}
In both equilibria, a player who has prepared the other action executes if the resulting payoff exceeds the discounted equilibrium payoff, and revises otherwise. Players never revise their plans on the equilibrium path.
\end{prop}
 
\begin{proof}
We verify both profiles in turn. Since $p,q\in(1/2,1)$, both signals occur with positive probability, and beliefs are derived by Bayes' rule.
 
\begin{itemize}
    \item \textit{The $(A,A)$ Equilibrium:} Let both players set $\alpha_1=\alpha_2=1$. Since the opponent prepares $A$ with certainty, the posterior beliefs are $\mu_i(A_{-i}\mid x_i,\hat A_i)=\mu_i(A_{-i}\mid x_i,\hat B_i)=1$ for both players $i\in\{1,2\}$. A player who prepared $A_i$ obtains $a$ from executing and at most $\delta a<a$ from revising, since no payoff exceeds $a$. Hence executing after either signal is optimal, and players never revise on the equilibrium path. A player who prepared $B_i$ obtains $d$ from executing, since the opponent executes, and $\delta a$ from revising, and his strategy prescribes the better of the two. Suppose Player 1 deviates by preparing $B_1$. He then obtains $\max\{d,\delta a\}<a$, so the deviation is unprofitable, and $(A,A)$ is a PBE.
 
    \item \textit{The $(B,B)$ Equilibrium:} Let both players set $\alpha_1=\alpha_2=0$. The posterior beliefs are $\mu_i(B_{-i}\mid x_i,\hat A_i)=\mu_i(B_{-i}\mid x_i,\hat B_i)=1$. A player who prepared $B_i$ obtains $b$ from executing and $\delta b<b$ from revising, so executing after either signal is optimal, and players never revise on the equilibrium path. A player who prepared $A_i$ obtains $c$ from executing, since the opponent executes, and $\delta b$ from revising, and his strategy prescribes the better of the two. Suppose Player 1 deviates by preparing $A_1$. He then obtains $\max\{c,\delta b\}<b$, so the deviation is unprofitable, and $(B,B)$ is a PBE.
\end{itemize}
\end{proof}
 
\begin{rem}[Off-Path Revisions]
Because strategies are stationary, beliefs about the opponent's next preparation are derived from $\alpha_{-i}^*$ after any revision, so that a revision is not interpreted as a signal about future play. With history-dependent strategies, Player 2 could interpret an unexpected revision in the $(B,B)$ equilibrium as a signal that Player 1 has switched to $A_1$ and best respond with $A_2$. This interpretation is fragile, however, since nothing prevents Player 1 from reverting to $B_1$ after a revision. Section \ref{Marginalizing Inferior Equilibria} shows that trembles eliminate $(B,B)$ also in non-stationary equilibria, without relying on such signaling.
\end{rem}
 
\begin{rem}[Never Executing]\label{Rem Never Executing}
In addition to the equilibria of Proposition \ref{Proposition Pure Equilibria}, the game admits a degenerate equilibrium. Since execution requires mutual consent, never executing is a best response to an opponent who never executes, so that the profile in which both players always revise is a Perfect Bayesian Equilibrium of the game without trembles, even though its payoff of zero is worse than every executed outcome. This equilibrium relies on weakly dominated strategies: never executing yields zero against every strategy of the opponent, whereas always executing yields at least zero, since all payoffs are positive, and a strictly positive payoff against an opponent who always executes. It disappears once trembles are introduced in Section \ref{Marginalizing Inferior Equilibria}, because a player facing a trembling opponent who prepares and executes $A$ strictly gains from executing, and none of our results relies on excluding it.
\end{rem}
 
\section{Toppling Inferior Equilibria}\label{Marginalizing Inferior Equilibria}
 
In this section, we show that cheap, precise scouting destabilizes the safe $(B,B)$ equilibrium by allowing agents to screen for cooperative preparations. As a result, the low-payoff outcome unravels in the presence of small trembles.

\begin{prop}\label{Prop Instability time pref}
Consider the pure-strategy equilibrium $(B,B)$. Suppose Player 2 makes occasional errors, planning and unconditionally executing action $A_2$ with probability $\varepsilon \in (0,1)$. Then, for any fixed tremble $\varepsilon > 0$, there exist threshold values $p^*(\varepsilon) \in (1/2, 1)$ and $\delta^*(\varepsilon) \in (0, 1)$ such that if $p > p^*$ and $\delta > \delta^*$, it is profitable for Player 1 to deviate and to plan action $A_1$ in every iteration and execute if and only if the scouting signal reads $\hat{A}_1$.
\end{prop}
 
\begin{proof}
Consider an off-path deviation where Player 1 plans $A_1$ in every iteration and chooses to execute in Stage 3 if and only if the observed signal is $\hat{A}_1$. 
 
To establish the ex-ante as well as the sequential rationality of Player 1's deviation, we must evaluate the ex-interim beliefs and expected payoffs after the scouting signal is received but before the execution decision is made. 
 
Using Bayes' rule, Player 1's ex-interim belief that Player 2 has prepared $A_2$ conditional on observing $\hat{A}_1$ is:
$$ \mu_1(A_2 \mid \hat{A}_1) = \frac{p \varepsilon}{p \varepsilon + (1-p)(1-\varepsilon)} $$ 
 
Conversely, upon observing $\hat{B}_1$, Player 1's ex-interim belief is:
$$ \mu_1(A_2 \mid \hat{B}_1) = \frac{(1-p) \varepsilon}{(1-p) \varepsilon + p(1-\varepsilon)} $$ 
 
Given these beliefs, Player 1's expected ex-interim payoff for choosing to execute is:
$$ E[U_1 \mid \hat{A}_1] = \mu_1(A_2 \mid \hat{A}_1) a + (1 - \mu_1(A_2 \mid \hat{A}_1)) c $$
$$ E[U_1 \mid \hat{B}_1] = \mu_1(A_2 \mid \hat{B}_1) a + (1 - \mu_1(A_2 \mid \hat{B}_1)) c $$
 
Given Player 1's screening strategy, taking into account that Player 2 never triggers a revision in the $(B,B)$ equilibrium, execution occurs in any given iteration with the probability of observing $\hat{A}_1$:
$$ \gamma = P(\hat{A}_1) = p\varepsilon + (1-p)(1-\varepsilon) $$
 
The ex-ante discounted present value $V_1$ of this infinite-horizon screening strategy is thus:
$$ V_1 = \sum_{t=0}^{\infty} \delta^t (1-\gamma)^t \gamma E[U_1 \mid \hat{A}_1] = \frac{p \varepsilon a + (1-p)(1-\varepsilon) c}{1 - \delta(1-\gamma)} $$
 
\textit{Sequential Rationality:} For the screening strategy to be sequentially rational at the ex-interim stage, Player 1 must prefer to execute upon observing $\hat{A}_1$ and revise upon observing $\hat{B}_1$. This requires comparing the expected immediate execution payoff to the discounted continuation value of revising, $\delta V_1$:
$$ E[U_1 \mid \hat{A}_1] > \delta V_1 > E[U_1 \mid \hat{B}_1] $$
 
Taking the limit ($p \to 1$):
$$ \lim_{p \to 1} \mu_1(A_2 \mid \hat{A}_1) = 1 \implies \lim_{p \to 1} E[U_1 \mid \hat{A}_1] = a $$
$$ \lim_{p \to 1} \mu_1(A_2 \mid \hat{B}_1) = 0 \implies \lim_{p \to 1} E[U_1 \mid \hat{B}_1] = c $$
$$ \lim_{p \to 1} \gamma = \varepsilon $$
$$ \lim_{p \to 1} V_1 = \frac{\varepsilon a}{1 - \delta(1-\varepsilon)} $$ 
 
Because $\varepsilon \in (0,1)$ and $\delta \in (0,1)$, it holds that $\frac{\varepsilon}{1 - \delta(1-\varepsilon)} < 1$, ensuring $a > V_1 > \delta V_1$. Furthermore, as $\delta \to 1$, we have $\delta V_1 \to V_1 \to a$, which exceeds $c$. Thus, for sufficiently high $p$ and $\delta$, the sequential rationality condition $E[U_1 \mid \hat{A}_1] > \delta V_1 > E[U_1 \mid \hat{B}_1]$ is satisfied.
 
\textit{Ex-ante Rationality:} In the $(B,B)$ equilibrium with trembles, Player 1 plans $B_1$ and executes unconditionally, which yields $\tilde b:=(1-\varepsilon)b+\varepsilon d$, since Player 2 executes $A_2$ with probability $\varepsilon$. Player 1 \textit{ex-ante} prefers the screening strategy if and only if its value exceeds $\tilde b$. Evaluating this inequality under $p\to1$:
$$ \frac{\varepsilon a}{1 - \delta(1-\varepsilon)} > \tilde b \iff \varepsilon a > \tilde b(1 - \delta + \delta \varepsilon) \iff \delta > \frac{\tilde b - \varepsilon a}{\tilde b(1 - \varepsilon)}, $$
where the threshold is below one because $a>\tilde b$.
 
The deviation must also satisfy sequential rationality at the ex-interim stage, i.e.\ $\delta V_1 > c$. Substituting $V_1 = \frac{\varepsilon a}{1-\delta(1-\varepsilon)}$, this condition reduces to $\delta > \frac{c}{\varepsilon a + c(1-\varepsilon)}$. Let
$$ \underline{\delta}(\varepsilon) \equiv \max\left\{\ \frac{\tilde b - \varepsilon a}{\tilde b(1 - \varepsilon)},\ \frac{c}{\varepsilon a + c(1-\varepsilon)}\right\} \in (0,1), $$
and fix any $\delta^*(\varepsilon)\in(\underline{\delta}(\varepsilon),1)$.
 
At $\delta=\delta^*$, the inequalities $V_1>\tilde b$ and $\delta V_1>E[U_1\mid\hat B_1]$ hold strictly in the limit $p\to1$. By continuity in $p$, there exists $p^*<1$ such that both hold for all $p>p^*$ at $\delta=\delta^*$. Since $V_1$ and $\delta V_1$ are increasing in $\delta$, and $\delta V_1=\frac{\delta\gamma}{1-\delta+\delta\gamma}E[U_1\mid\hat A_1]<E[U_1\mid\hat A_1]$ for every $\delta<1$, all conditions continue to hold for all $\delta>\delta^*$. Thus, Player 1 prefers to deviate to screening, rendering $(B,B)$ unstable.
\end{proof}
 
Proposition \ref{Prop Instability time pref} indicates that whenever the cost of revising plans is low ($\delta \to 1$) and scouting precision is high ($p \to 1$), players have a strong incentive to screen for favorable engagements even when off-equilibrium cooperative play is extremely rare. The following result extends this screening argument from unconditionally executing trembles to arbitrary stationary strategies of the opponent.
 
\begin{prop}[Vanishing Threshold]\label{Prop Vanishing Threshold}
Let Player $i$ face a stationary strategy of the opponent, and let
$$\bar\varepsilon_{-i}:=\alpha_{-i}\big[\pi_{-i}\,\chi_{-i}(A_{-i},\hat A_{-i})+(1-\pi_{-i})\,\chi_{-i}(A_{-i},\hat B_{-i})\big]$$
be the probability that the opponent prepares $A_{-i}$ and is willing to execute, given that Player $i$ has prepared $A_i$ (if the opponent randomizes in Stage 3, as under trembles, $\chi_{-i}$ denotes his probability of executing). Let
$$\hat\eta_i:=\frac{d\,(1-\delta\pi_i)}{\pi_i\,(a-\delta d)} .$$
If $\bar\varepsilon_{-i}>\hat\eta_i$, then every best response of Player $i$ prepares $A_i$ in every iteration. Moreover, $\hat\eta_i\to0$ as $(\pi_i,\delta)\to(1,1)$.
\end{prop}
 
\begin{proof}
Consider the strategy in which Player $i$ prepares $A_i$ in every iteration and executes if and only if he observes $\hat A_i$. In each iteration, the profile $(A,A)$ is executed with probability $\pi_i\bar\varepsilon_{-i}$, since Player $i$ observes $\hat A_i$ with probability $\pi_i$ when the opponent prepared $A_{-i}$, independently of the opponent's signal. If the opponent prepared $B_{-i}$, Player $i$ executes only after a false signal, i.e.\ with probability at most $1-\pi_i$. Hence the probability of an execution in a given iteration is at most $\pi_i\bar\varepsilon_{-i}+1-\pi_i$. Since all payoffs are positive, the expected payoff collected in a given iteration is at least $a\pi_i\bar\varepsilon_{-i}$. Because the opponent's strategy is stationary and signals are drawn independently across iterations, the value of this strategy is at least
$$\Lambda_i(\bar\varepsilon_{-i}):=\frac{a\,\pi_i\bar\varepsilon_{-i}}{1-\delta+\delta(\pi_i\bar\varepsilon_{-i}+1-\pi_i)} .$$
Rearranging shows that $\Lambda_i(\bar\varepsilon_{-i})>d$ if and only if $\bar\varepsilon_{-i}>\hat\eta_i$. Let $V^*$ denote Player $i$'s best-response value, so that $V^*\ge\Lambda_i(\bar\varepsilon_{-i})>d$. If Player $i$ prepares $B_i$ in some iteration, every execution in that iteration pays him $b$ or $d$, and a revision yields at most $\delta V^*$. Hence preparing $B_i$ yields at most $\max\{d,\delta V^*\}<V^*$, so it is not part of any best response. Finally, $1-\delta\pi_i\to0$ and $a-\delta d\to a-d>0$ as $(\pi_i,\delta)\to(1,1)$, so $\hat\eta_i\to0$.
\end{proof}

Proposition \ref{Prop Vanishing Threshold} has immediate implications for the mixed equilibria of the game and for the robustness of its pure equilibria.

\begin{coro}\label{Coro Complete Strategy}
Consider any stationary equilibrium in which strategies are monotone, i.e.\ $\chi_i(A_i,\hat A_i)\ge\chi_i(A_i,\hat B_i)$, and in which players mix between preparing $A$ and $B$. Let $\omega_i:=\alpha_i\big[P(\hat A_i)\chi_i(A_i,\hat A_i)+P(\hat B_i)\chi_i(A_i,\hat B_i)\big]$ denote the ex-ante probability that Player $i$ prepares $A_i$ and is willing to execute. As $p,q,\delta\to1$, the probabilities $\omega_i$ either go to zero or the mixed equilibrium ceases to exist.
\end{coro}
\begin{proof}
In a mixed equilibrium, Player $i$ prepares $B_i$ with positive probability, so preparing $B_i$ is part of a best response, and Proposition \ref{Prop Vanishing Threshold} implies $\bar\varepsilon_{-i}\le\hat\eta_i$. By monotonicity, $\chi_{-i}(A_{-i},\hat A_{-i})\ge\chi_{-i}(A_{-i},\hat B_{-i})$, so that $\omega_{-i}\le\alpha_{-i}\,\chi_{-i}(A_{-i},\hat A_{-i})\le\bar\varepsilon_{-i}/\pi_{-i}$. Hence $\omega_{-i}\le\hat\eta_i/\pi_{-i}\to0$ as $p,q,\delta\to1$.
\end{proof}
 
We say that the opponent trembles towards $x\in\{A,B\}$ with probability $\varepsilon$ if, in each iteration and independently across iterations, he prepares $x$ and executes after either signal with probability $\varepsilon$, and follows his equilibrium strategy otherwise.
 
\begin{coro}\label{Coro BB Unstable}
Fix $\varepsilon\in(0,1)$. For $p,q,\delta$ sufficiently close to one, neither $(B,B)$ nor any stationary mixed equilibrium survives trembles of size $\varepsilon$ towards $A$: against such trembles, every best response prepares $A$. Moreover, for all $p,q,\delta$, the $(A,A)$ equilibrium survives trembles of size $\varepsilon<\frac{a-d}{a-c}$ towards $B$: every best response prepares $A$.
\end{coro}
\begin{proof}
Trembles towards $A$ imply $\bar\varepsilon_{-i}\ge\varepsilon>\hat\eta_i$ for $p,q,\delta$ close to one, so the first claim follows from Proposition \ref{Prop Vanishing Threshold}. Against trembles towards $B$ in the $(A,A)$ equilibrium, preparing $A$ and executing after either signal yields $(1-\varepsilon)a+\varepsilon c>d$, so that the best-response value satisfies $V^*>d$. As in the proof of Proposition \ref{Prop Vanishing Threshold}, preparing $B$ yields at most $\max\{d,\delta V^*\}<V^*$.
\end{proof}
 
\begin{rem}
The analysis so far concerns perturbations that are redrawn independently in each iteration, while the two players remain the same. Thus, whether a player trembled in one iteration carries no information about his behavior in later iterations.
\end{rem}
 
We now show that trembles eliminate the safe outcome also in non-stationary equilibria. Suppose that in each iteration, independently across players and iterations, each player trembles with probability $\varepsilon\in(0,1)$, in which case he prepares $A$ and executes if and only if he observes $\hat A$.\footnote{If trembles executed unconditionally on both sides, a player's own trembles would end play at $(A,B)$ and dilute the value of screening.} Strategies may depend on the history of a player's own preparations, signals, and Stage-3 choices. Perfect Bayesian Equilibrium is defined as in Definition \ref{Def PBE}, with strategies and beliefs depending on histories and trembles treated as moves of nature.
 
Trembles are mechanical: for any history, the opponent trembles with probability $\varepsilon$, and then prepares $A$ and executes after $\hat A$. Hence no history-dependent strategy can deprive a screening player of the opportunity to meet a cooperative preparation, and the screening argument of Proposition \ref{Prop Vanishing Threshold} extends to all equilibria.
 
\begin{prop}[Selection in Every Equilibrium]\label{Prop Selection}
Let $\Lambda_i(\eta):=\frac{a\,\pi_i\eta}{1-\delta+\delta(\pi_i\eta+1-\pi_i)}$ as in the proof of Proposition \ref{Prop Vanishing Threshold}. In every Perfect Bayesian Equilibrium, stationary or not, the equilibrium payoff of each Player $i$ satisfies $U_i^*\ge\Lambda_i(\pi_{-i}\varepsilon)$, and the probability that play ends in $(A,A)$ is at least $\frac{\Lambda_i(\pi_{-i}\varepsilon)-d}{a-d}$. Both bounds converge to $a$ and $1$, respectively, as $p,q,\delta\to1$.
\end{prop}
 
\noindent The proof is given in Appendix \ref{Appendix Proofs}. Proposition \ref{Prop Selection} shows that screening selects the payoff-dominant outcome without any restriction to stationary strategies. In particular, the safe outcome $(B,B)$ cannot be sustained by threatening to punish a player whose screening triggers revisions. Moreover, equilibria exist: since trembling players also prepare $A$, the profile in which non-trembling players prepare $A$, execute after either signal, and make optimal Stage-3 choices after preparing $B$ is a Perfect Bayesian Equilibrium for all parameter values.

\section{Diagrammatic Interpretation}\label{Section Basins}
 
Figure~\ref{fig:vanishing_basin_space} summarizes Proposition \ref{Prop Vanishing Threshold} and Corollaries \ref{Coro Complete Strategy} and \ref{Coro BB Unstable}. The axes show the probabilities $\bar\varepsilon_1$ and $\bar\varepsilon_2$ with which each player prepares $A$ and is willing to execute, given that the opponent has prepared $A$. By Proposition \ref{Prop Vanishing Threshold}, preparing $B$ can be a best response for Player $i$ only if $\bar\varepsilon_{-i}\le\hat\eta_i$. Hence, both players can prefer $B$ only in the red region near the origin, whereas in the grey regions preparing $A$ is the unique best response of one player, and in the blue region of both. As scouting precision increases ($p,q\to1$) and revisions become cheap ($\delta\to1$), the thresholds $\hat\eta_i$ vanish and the red region shrinks towards the origin. Put differently, by selectively executing only those coordination attempts that yield the payoff-dominant outcome $a$, players' screening efforts compress the mixed equilibria, which lie in the red region, towards the origin, where $(B,B)$ lies, while any fixed tremble towards $A$ moves play out of this region.
 
\begin{figure}[htbp]
\centering
\begin{tikzpicture}[scale=0.85, every node/.style={font=\small}]
    \colorlet{basinB}{red!15}
    \colorlet{basinA}{blue!15}
    \colorlet{trans}{gray!10}
 
    \fill[basinB] (0,0) rectangle (1.5,1.5);
    \fill[basinA] (1.5,1.5) rectangle (10,10);
    
    \fill[trans] (0,1.5) rectangle (1.5,10);
    \fill[trans] (1.5,0) rectangle (10,1.5);
 
    \draw[thick, ->] (0,0) -- (11,0) node[right] {$\bar\varepsilon_1$};
    \draw[thick, ->] (0,0) -- (0,11) node[above] {$\bar\varepsilon_2$};
    \draw[thick] (10,0) -- (10,10) -- (0,10);
    
    \draw[thick] (10,0) -- (10,-0.2) node[below] {$1$};
    \draw[thick] (0,10) -- (-0.2,10) node[left] {$1$};
    \draw[thick] (1.5,0) -- (1.5,-0.2) node[below] {$\hat\eta_2$};
    \draw[thick] (0,1.5) -- (-0.2,1.5) node[left] {$\hat\eta_1$};
 
    \draw[dashed, thick, darkgray] (1.5,0) -- (1.5,10);
    \draw[dashed, thick, darkgray] (0,1.5) -- (10,1.5);
 
    \filldraw[black] (0,0) circle (3pt) node[below left] {$(B,B)$};
    \filldraw[black] (10,10) circle (3pt) node[above right] {$(A,A)$};
 
    \node[red!80!black, align=center, font=\footnotesize] at (0.75, 0.75) {Mixed\\ eq.};
    \node[blue!80!black, align=center, font=\large] at (6,6) {$A$ is the unique best\\ response of both players};
 
    \draw[->, thick, black] (1.5, 1.5) -- (0.4, 0.4) node[midway, sloped, above, font=\scriptsize] {$p, q, \delta \to 1$};
 
\end{tikzpicture}
\caption{Best responses as functions of the players' willingness to execute $A$. Preparing $B$ can be a best response of Player $i$ only if $\bar\varepsilon_{-i}\le\hat\eta_i$, so that both players can prefer $B$ only in the red region, which contains all mixed equilibria. As $p,q\to1$ and $\delta\to1$, the thresholds $\hat\eta_1,\hat\eta_2$ vanish, the red region shrinks towards the origin, where $(B,B)$ lies, while any fixed tremble towards $A$ moves play into the region where $A$ is a best response.}
\label{fig:vanishing_basin_space}
\end{figure}
 
\section{Forced Executions}\label{Section Unilateral Execution}
 
This section revisits our previous selection argument under the alternative assumption that players can unilaterally force the execution of plans even when their opponents seek a revision.

\begin{definition}[Strategy with Forced Execution and History]\label{Def Unilateral Strategies}
Let $\mathcal{H}$ denote the set of all possible histories of a player's own preparations, signals, and Stage-3 choices in prior revised rounds, with $h^t \in \mathcal{H}$ denoting the history prior to iteration $t$. A strategy for Player $i \in \{1,2\}$ under the unilateral-execution protocol is a pair $s_i = (\alpha_i, \chi_i)$, where:
\begin{enumerate}
    \item $\alpha_i: \mathcal{H} \to [0,1]$ specifies the probability of preparing action $A_i$ in Stage 1 as a function of history $h^t$.
    \item $\chi_i : \mathcal{H} \times \{A_i, B_i\} \times \{\hat{A}_i, \hat{B}_i\} \to \{0,1,2\}$ specifies, given preparation $x_i \in \{A_i, B_i\}$ and observed signal $\hat{y}_i \in \{\hat{A}_i, \hat{B}_i\}$, one of three Stage-3 choices (we suppress $h^t$):
    \begin{itemize}
        \item $\chi_i(x_i,\hat{y}_i) = 0$ \textbf{(Revise):} Player $i$ requests a return to Stage 1.
        \item $\chi_i(x_i,\hat{y}_i) = 1$ \textbf{(Execute):} Player $i$ is willing to execute the current plans $(x_1,x_2)$.
        \item $\chi_i(x_i,\hat{y}_i) = 2$ \textbf{(Force):} Player $i$ unilaterally forces execution.
    \end{itemize}
\end{enumerate}
\end{definition}
 
\textit{Aggregation of Choices:} The Stage-3 outcome of the game is thus determined by the aggregation of choices. The current plans $(x_1, x_2)$ are executed if either $\max\{\chi_1, \chi_2\} = 2$ or $\min\{\chi_1, \chi_2\} = 1$. Otherwise, the game returns to Stage 1.
 
\textit{Environment.} As in Proposition \ref{Prop Selection}, in each iteration, independently across players and iterations, each player trembles with probability $\varepsilon\in(0,1)$, in which case he prepares $A$, executes if and only if he observes $\hat A$, and never forces. A player who forces execution incurs a cost $\tau\ge0$. Perfect Bayesian Equilibrium is defined as in Section \ref{Marginalizing Inferior Equilibria}, with Stage-3 choices in $\{0,1,2\}$.
 
\begin{prop}[Costless Forcing]\label{Prop Unilateral}
Suppose that $\tau=0$, and let $\alpha^*=\frac{b-c}{(a-d)+(b-c)}$ be as in Proposition \ref{Prposition 1 Primitive Stag Hunt}. For all $p,q\in(1/2,1)$ and $\delta\in(0,1)$:
\begin{enumerate}
\item[(i)] There is a stationary Perfect Bayesian Equilibrium in which players prepare $A$ and force execution when they do not tremble. Play always ends in $(A,A)$; plans are revised only if both players tremble.
\item[(ii)] If $\varepsilon\le\alpha^*$, there is a stationary Perfect Bayesian Equilibrium in which players prepare $B$ and force execution when they do not tremble. Play ends in $(B,B)$ in the first iteration unless a player trembles.
\item[(iii)] If $\varepsilon<\alpha^*$, there is a stationary Perfect Bayesian Equilibrium in which players, when they do not tremble, prepare $A$ with probability
$$\beta:=\frac{(b-c)-\frac{\varepsilon}{1-\varepsilon}(a-d)}{(a-d)+(b-c)}\in(0,\alpha^*)$$
and force execution.
\end{enumerate}
In particular, precise scouting and cheap revisions do not select among the equilibria of the stage game.
\end{prop}
 
\begin{proof}
Fix $\beta\in[0,1]$ and consider the stationary profile in which a player who does not tremble prepares $A$ with probability $\beta$ and forces at every Stage-3 information set. Beliefs are derived by Bayes' rule from this profile. Let $U_i$ denote Player $i$'s ex-ante equilibrium payoff and $W_i:=\delta U_i$ his continuation value after a revision.
 
\emph{Payoffs.} A player who does not tremble, prepares $x$, and forces obtains $u(x,y)$ immediately, where $y$ is the opponent's preparation. A non-trembling opponent prepares $A$ with probability $\beta$, and a trembling opponent prepares $A$ with certainty. Hence
$$v_A:=(1-\varepsilon)\big[\beta a+(1-\beta)c\big]+\varepsilon a,\qquad v_B:=(1-\varepsilon)\big[\beta d+(1-\beta)b\big]+\varepsilon d$$
are the payoffs from preparing $A$ and $B$, respectively, followed by forcing.
 
\emph{Stage 3.} Consider a player who does not tremble and prepared $x$. If the opponent does not tremble, the opponent forces, so the current plans are executed whatever the player does. If the opponent trembles, the opponent prepared $A$ and never forces, so forcing yields $u(x,A)$, whereas executing or revising yields either $u(x,A)$ or the continuation value $W_i$. Forcing is therefore optimal, whatever the player's beliefs, whenever $u(x,A)\ge W_i$. For $x=A$, this holds because $u(A,A)=a>\delta U_i=W_i$. For $x=B$, it holds whenever $U_i\le d$, which we verify below in cases (ii) and (iii). In case (i), $B$ is prepared only after a deviation, and we let a player who prepared $B$ choose any optimal Stage-3 action.
 
\emph{Case (i), $\beta=1$.} Every type of the opponent prepares $A$, so preparing $A$ and forcing yields $v_A=a$, the highest payoff in the game. Preparing $B$ yields at most $\max\{d,W_i\}<a$, since every execution then pays at most $d$ and a revision pays $W_i\le\delta a$. Hence preparing $A$ is optimal.
 
\emph{Case (ii), $\beta=0$.} Here $v_A=(1-\varepsilon)c+\varepsilon a$ and $v_B=(1-\varepsilon)b+\varepsilon d$, so $v_B\ge v_A$ if and only if $\varepsilon(a-d)\le(1-\varepsilon)(b-c)$, i.e.\ if and only if $\varepsilon\le\alpha^*$. In iterations in which a player trembles, he obtains at most $(1-\varepsilon)c+\varepsilon a=v_A$, since against a non-trembling opponent the plans $(A,B)$ are executed immediately and against a trembling opponent he obtains at most $a$. Hence $U_i\le(1-\varepsilon)v_B+\varepsilon v_A\le v_B\le d$, and forcing is optimal at every Stage-3 information set.
 
\emph{Case (iii), $\beta\in(0,1)$.} Mixing requires $v_A=v_B$, i.e.
$$(1-\varepsilon)\big[\beta(a-d)-(1-\beta)(b-c)\big]+\varepsilon(a-d)=0,$$
which is solved by the stated $\beta$. We have $\beta<\alpha^*$, and $\beta>0$ if and only if $\frac{\varepsilon}{1-\varepsilon}<\frac{b-c}{a-d}=\frac{\alpha^*}{1-\alpha^*}$, i.e.\ if and only if $\varepsilon<\alpha^*$. As in case (ii), a trembling player obtains at most $v_A$, so that $U_i\le v_B\le d$, and forcing is optimal.
 
In all three cases, players who do not tremble best respond at every information set. Hence each profile is a stationary Perfect Bayesian Equilibrium.
\end{proof}
 
Proposition \ref{Prop Unilateral} shows that the selection result of Proposition \ref{Prop Selection} depends on mutual consent. Screening requires the ability to refuse the execution of plans after an unfavorable signal. Against an opponent who forces, this veto is lost: a player who prepares $A$ and faces an opponent who prepared $B$ is executed at $(A,B)$ and obtains $c$. The game then collapses to the stage game, whose three equilibria reappear for all signal precisions and discount factors. If forcing is sufficiently costly, however, the veto is restored: a player who prepared $B$ never finds it worthwhile to force, since the most he can gain from forcing, $d$, falls short of its cost.
 
\begin{prop}[Costly Forcing]\label{Prop Costly Forcing}
Suppose that $\tau>d$. Then, in every Perfect Bayesian Equilibrium, stationary or not, the equilibrium payoff of each Player $i$ satisfies $U_i^*\ge\Lambda_i(\pi_{-i}\varepsilon)$, and the probability that play ends in $(A,A)$ is at least $\frac{\Lambda_i(\pi_{-i}\varepsilon)-d}{a-d}$. Both bounds converge to $a$ and $1$, respectively, as $p,q,\delta\to1$.
\end{prop}
 
\noindent The proof is given in Appendix \ref{Appendix Proofs}. The equilibrium profile described after Proposition \ref{Prop Selection}, without forcing, remains a Perfect Bayesian Equilibrium whenever $\varepsilon a\le d$. Together, Propositions \ref{Prop Unilateral} and \ref{Prop Costly Forcing} show that selection hinges on how costly it is to override the opponent's veto.
 
\begin{rem}[Cheap Forcing]\label{Rem Cheap Forcing}
If forcing is costly but $\tau$ is small, the safe equilibrium can survive. For all but very small $\tau>0$, the profiles of Proposition \ref{Prop Unilateral} are no longer equilibria, since forcing against an opponent who forces as well costs $\tau$ and gains nothing. However, players who prepare $B$ may still force after observing $\hat A$, which makes a trembling opponent more likely, in order to obtain $d-\tau$ rather than a discounted fresh game. For example, with $a=4$, $d=3$, $b=2.9$, $c=0.1$, $\varepsilon=0.01$, $p=q=\delta=0.99$ and $\tau=0.01$, the stationary profile in which players who do not tremble prepare $B$, force after $\hat A$, and execute after $\hat B$ is a Perfect Bayesian Equilibrium. A player who screens against it faces an opponent who forces after $\hat A$, and is executed at $(A,B)$.
\end{rem}
 
\begin{rem}[Pure Coordination]\label{Rem Pure Coordination}
Suppose that the payoff $d$ from playing $B$ against $A$ is replaced by $c$, so that both miscoordinated profiles pay $c$ and $a>b>c$. Our selection results continue to hold, and all thresholds become less demanding. In Proposition \ref{Prop Instability time pref}, the benchmark payoff of the safe equilibrium falls to $(1-\varepsilon)b+\varepsilon c$. In Propositions \ref{Prop Vanishing Threshold}, \ref{Prop Selection}, and \ref{Prop Costly Forcing}, $d$ enters only as the highest payoff that a player who prepared $B$ can obtain, which is now $b$; hence the threshold $\hat\eta_i$ decreases, the lower bound on the probability that play ends in $(A,A)$ increases, and costly forcing restores selection whenever $\tau>b$. More importantly, forcing loses its bite. In Proposition \ref{Prop Unilateral}, a player who prepared $B$ forces the execution of a mismatch in order to obtain $d$. If this payoff is $c$ instead, forcing a mismatch pays less than the continuation value when $\varepsilon$ is small and $\delta$ is close to one, so that the safe equilibrium of Proposition \ref{Prop Unilateral} (ii) no longer exists. At the same time, with $d=c$ the payoff-dominant equilibrium is also risk dominant, since $a-c>b-c$, so that standard selection criteria already favor $(A,A)$. The stag hunt with $d>b$ is the more demanding case: there, screening selects $(A,A)$ even if $(B,B)$ is risk dominant, i.e.\ even if $a-d<b-c$.
\end{rem}
 
\section{Discussion}\label{Discussion}
 
The present paper stresses that (i) every action requires preparation and that (ii) these preparations are monitored, such that (iii) players partially know each other's contemporaneous actions once they decide to execute their plans. This modeling approach is not only plausible in many contexts but also avoids the ``chicken-and-egg'' problems between signals and actions that arise in standard simultaneous-move games.
 
Strategically, monitoring preparations creates an incentive for players to screen for cooperative behavior: a player prepares the cooperative action and uses his veto over execution after an unfavorable signal. Once monitoring is precise and revisions are cheap, screening drives the probability with which players prepare $A$ and are willing to execute it in mixed equilibria to zero, as in the safe equilibrium $(B,B)$, and small trembles cause $(B,B)$ to unravel, so that every equilibrium ends in the payoff-dominant outcome $(A,A)$ with a probability close to one, even when $(B,B)$ is risk dominant. This selection relies on the veto: if players can force the execution of plans at no cost, all equilibria of the stage game reappear, whereas sufficiently costly forcing restores selection. More generally, the observation that screening allows players to target the cooperative part of their opponent's play, such that mixed equilibria collapse toward the safe equilibrium, may prove useful in other games in which pure and mixed equilibria coexist.
 
Several questions remain open. First, we have focused on perturbations that are redrawn in each iteration; with persistent perturbations, observed behavior reveals information about a player's type, and players may force executions to prevent such learning. Second, for forcing costs below the threshold of Proposition \ref{Prop Costly Forcing}, the safe equilibrium can survive, and a full characterization of the role of forcing costs is left for future research. Finally, it would be interesting to study whether screening also sustains cooperation in games such as the prisoner's dilemma.
 

\newpage
\appendix
\section{Proofs of the Results on Non-Stationary Equilibria}\label{Appendix Proofs}
 
The following lemma extends the screening bound in the proof of Proposition \ref{Prop Vanishing Threshold} to history-dependent strategies of the opponent.
 
\begin{lemma}[Screening Bound]\label{Lemma Screening}
Fix a strategy of Player $-i$, which may depend on the history. For a history $h$ at the beginning of an iteration, let $\bar\varepsilon_{-i}(h)$ denote the probability that Player $-i$ prepares $A_{-i}$ and is willing to execute, given that Player $i$ has prepared $A_i$, where $h$ denotes the joint history of both players. Under the protocol of Section \ref{Section Unilateral Execution}, ``willing to execute'' includes forcing, and Player $-i$ is assumed not to force after preparing $B_{-i}$. If $\bar\varepsilon_{-i}(h)\ge\eta>0$ at every history $h$, and Player $i$ prepares $A_i$ in every iteration and executes if and only if he observes $\hat A_i$, then his continuation value at every history is at least
$$\Lambda_i(\eta):=\frac{a\,\pi_i\eta}{1-\delta+\delta(\pi_i\eta+1-\pi_i)}\ \longrightarrow\ a\qquad\text{as }(\pi_i,\delta)\to(1,1).$$
\end{lemma}
 
\begin{proof}
Consider an iteration beginning at history $h$. The profile $(A,A)$ is executed with probability at least $\pi_i\bar\varepsilon_{-i}(h)$, since Player $i$ observes $\hat A_i$ with probability $\pi_i$ when Player $-i$ prepared $A_{-i}$, independently of Player $-i$'s signal. If Player $-i$ prepared $B_{-i}$, Player $i$ executes only after a false signal, i.e.\ with probability at most $1-\pi_i$. Hence plans are revised with probability at least $\pi_i(1-\bar\varepsilon_{-i}(h))$, and, since all payoffs are positive, the expected payoff collected in this iteration is at least $a\pi_i\bar\varepsilon_{-i}(h)$. Let $V(h)$ denote Player $i$'s continuation value at $h$ and $L:=\inf_hV(h)\in[0,a]$. Then
$$V(h)\ \ge\ a\pi_i\bar\varepsilon_{-i}(h)+\delta\pi_i\big(1-\bar\varepsilon_{-i}(h)\big)L\ \ge\ a\pi_i\eta+\delta\pi_i(1-\eta)L,$$
where the second inequality holds because the middle expression is increasing in $\bar\varepsilon_{-i}(h)$, since $a\ge\delta L$. Taking the infimum over $h$ yields $L\ge a\pi_i\eta+\delta\pi_i(1-\eta)L$, i.e.\ $L\ge\Lambda_i(\eta)$.
\end{proof}
 
\begin{proof}[Proof of Proposition \ref{Prop Selection}]
Player $i$ can deviate to preparing $A_i$ in every iteration and executing if and only if he observes $\hat A_i$; in iterations in which he trembles, this coincides with his trembling behavior. In every iteration, independently of the history, Player $-i$ trembles with probability $\varepsilon$ and then prepares $A_{-i}$ and executes after $\hat A_{-i}$, which he observes with probability $\pi_{-i}$ given that Player $i$ prepared $A_i$. Hence $\bar\varepsilon_{-i}(h)\ge\pi_{-i}\varepsilon$ at every history, whatever Player $-i$'s strategy. By Lemma \ref{Lemma Screening}, the deviation yields at least $\Lambda_i(\pi_{-i}\varepsilon)$, and an equilibrium payoff cannot fall short of a feasible deviation. Finally, Player $i$ obtains $a$ at $(A,A)$, at most $d$ at every other profile, and zero if plans are never executed, so that $U_i^*\le aP+d(1-P)$, where $P$ is the probability that play ends in $(A,A)$. This yields the bound on $P$.
\end{proof}
 
\begin{proof}[Proof of Proposition \ref{Prop Costly Forcing}]
Consider a player who prepared $B$. Forcing yields him $u(B,y)-\tau\le d-\tau<0$. Executing instead yields $u(B,y)>0$ if the opponent executes or forces, and the continuation value, which is nonnegative, if the opponent revises. Hence executing is strictly better than forcing, whatever the player's beliefs, and no player who prepared $B$ forces in any Perfect Bayesian Equilibrium.
 
Now let Player $i$ deviate to preparing $A_i$ in every iteration, executing if and only if he observes $\hat A_i$, and never forcing; in iterations in which he trembles, this coincides with his trembling behavior. Consider an iteration beginning at any history. As in the proof of Proposition \ref{Prop Selection}, Player $-i$ trembles with probability $\varepsilon$, independently of the history, and then prepares $A_{-i}$ and executes after $\hat A_{-i}$, which he observes with probability $\pi_{-i}$. Hence $\bar\varepsilon_{-i}(h)\ge\pi_{-i}\varepsilon$ at every history, where, as in Lemma \ref{Lemma Screening}, willingness to execute includes forcing. Since Player $-i$ does not force after preparing $B_{-i}$, Lemma \ref{Lemma Screening} applies with $\eta=\pi_{-i}\varepsilon$, so the deviation yields at least $\Lambda_i(\pi_{-i}\varepsilon)$. Since an equilibrium payoff cannot fall short of a feasible deviation, $U_i^*\ge\Lambda_i(\pi_{-i}\varepsilon)$. Finally, Player $i$ obtains $a$ at $(A,A)$, at most $d$ at every other profile, and zero if plans are never executed, and forcing costs only lower his payoff. Hence $U_i^*\le aP+d(1-P)$, where $P$ is the probability that play ends in $(A,A)$, which yields the bound on $P$.
\end{proof}

\newpage
 
\addcontentsline{toc}{section}{References}
\markboth{References}{References}
\bibliographystyle{apalike}
\bibliography{References}
 
\end{document}